\documentclass[a4paper,UKenglish,cleveref,autoref,thm-restate]{lipics-v2021}

\pdfoutput=1 
\hideLIPIcs  

\usepackage{graphicx}
  \graphicspath{{./pics/}}
  \DeclareGraphicsExtensions{.pdf}
\usepackage{xspace}
\usepackage{outlines}
\usepackage{booktabs}
  \newcommand{\ra}[1]{\renewcommand{\arraystretch}{#1}}

\usepackage[textwidth=1.9cm,color=green!10,textsize=footnotesize]{todonotes}

\newcommand{\eps}{\varepsilon}
\DeclareMathOperator{\A}{\mathcal A}

\newcommand{\complclass}[1]{\textsc{#1}\xspace}
\newcommand{\PSpace}{\complclass{PSpace}}
\newcommand{\PTime}{\complclass{PTime}}
\newcommand{\NP}{\complclass{NP}}

\newtheorem{hypo}[theorem]{Hypothesis}

\title{Speed Me up if You Can:\\ Conditional Lower Bounds on Opacity Verification}

\titlerunning{Speed Me up if You Can: Conditional Lower Bounds on Opacity Verification}

\author{Ji\v{r}\'{i} Balun}{Faculty of Science, Palacky University Olomouc, Czechia}{jiri.balun01@upol.cz}{https://orcid.org/0000-0003-2332-9354}{Supported by IGA PrF 2023 026.}

\author{Tom{\' a}{\v s}~Masopust}{Faculty of Science, Palacky University Olomouc, Czechia}{tomas.masopust@upol.cz}{https://orcid.org/0000-0001-9282-758X}{}

\author{Petr Osi\v{c}ka}{Faculty of Science, Palacky University Olomouc, Czechia}{petr.osicka@upol.cz}{https://orcid.org/0000-0002-7495-7724}{}

\authorrunning{J. Balun, T. Masopust, and P. Osi\v{c}ka} 

\Copyright{Ji\v{r}\'{i} Balun, Tom{\' a}{\v s}~Masopust, and Petr Osi\v{c}ka} 



\nolinenumbers 

\EventEditors{John Q. Open and Joan R. Access}
\EventNoEds{2}
\EventLongTitle{42nd Conference on Very Important Topics (CVIT 2016)}
\EventShortTitle{CVIT 2016}
\EventAcronym{CVIT}
\EventYear{2016}
\EventDate{December 24--27, 2016}
\EventLocation{Little Whinging, United Kingdom}
\EventLogo{}
\SeriesVolume{42}
\ArticleNo{23}

\begin{document}

\maketitle

\begin{abstract}
  Opacity is a property of privacy and security applications asking whether, given a system model, a passive intruder that makes online observations of system's behaviour can ascertain some ``secret'' information of the system. Deciding opacity is a \PSpace-complete problem, and hence there are no polynomial-time algorithms to verify opacity under the assumption that \PSpace differs from \PTime. This assumption, however, gives rise to a question whether the existing exponential-time algorithms are the best possible or whether there are faster, sub-exponential-time algorithms. We show that under the (Strong) Exponential Time Hypothesis, there are no algorithms that would be significantly faster than the existing algorithms. As a by-product, we obtained a new conditional lower bound on the time complexity of deciding universality (and therefore also inclusion and equivalence) for nondeterministic finite automata.
\end{abstract}

\section{Introduction}
  In privacy and security applications and communication protocols, it is desirable to keep some information about the system or its behaviour secret. Such requirements put additional restrictions on the information flow of the system and have widely been discussed in the literature as various properties, including {\em anonymity}~\cite{Dingledine,ReiterR98,SchneiderS96}, {\em non-interference}~\cite{BestDG2010,BusiG09,Hadj-Alouane05,SabelfeldM03}, {\em secrecy}~\cite{Alur2006,BadouelBBCD07,BryansKR04}, {\em security}~\cite{Focardi94}, {\em perfect security}~\cite{ZakinthinosL97}, and {\em opacity}~\cite{Hadjicostis2020,Mazare04}.
  {\em Anonymity\/} is the property to preserve secrecy of identity of actions; for instance, web servers should not be able to learn the true source of a request.
  {\em Non-interference\/} asks whether, give two input states of the system that share the same values of specified variables, the behaviors of the system started from these states are indistinguishable by the observer under the observation of the specified variables.
  {\em Secrecy\/} expresses whether an observer can ever find out that a trajectory of the system belongs to a set of secret trajectories, and {\em perfect security\/} requires that an observer that knows the set of all trajectories of the system cannot deduce any information about occurrences of high-security events by observing low-security events.

  In this paper, we are interested in (various types of) opacity, which in a sense generalizes the other mentioned properties; namely, the properties above can be verified by reduction to opacity. More specifically, Alur et al.~\cite{Alur2006} have shown that {\em secrecy\/} captures {\em non-interference\/} and {\em perfect security}, and Lin~\cite{Lin2011} has provided an extensive discussion and comparison of all these properties. He has in particular shown that {\em anonymity\/} and {\em secrecy}, and the properties of {\em observability}~\cite{LinW88}, {\em diagnosability}~\cite{Lin94,SampathSLST95}, and {\em detectability}~\cite{Masopust18,MasopustY19,ShuLY07} of discrete-event systems are special cases of opacity.
  Wu et al.~\cite{WuRRLS18} and Góes et al.~\cite{Goes2018} discuss applications of opacity in location privacy, and Wintenberg et al.~\cite{WintenbergBLO22} apply opacity in contact tracing.

  Given a system model, {\em opacity\/} asks whether a malicious passive observer (an intruder) with a complete knowledge of the structure of the model can ascertain some ``secret'' information of the system by making online incomplete observations of system's behaviour. The secret information is modeled either as a set of states or as a set of behaviours of the system. Based on the incomplete observations, the intruder estimates the state/behaviour of the system, and the system is opaque if for every secret state/behaviour of the system, there is a non-secret state/behaviour of the system that looks the same to the intruder.
  If the secret information is given as a set of states, we talk about state-based opacity~\cite{BryansKMR08,Bryans2005}, whereas if the secret information is given as a set of secret behaviours (a language), we talk about language-based opacity~\cite{Badouel2007,Dubreil2008}.
  Several notions of opacity have been discussed in the literature for systems modeled by automata and Petri nets; see Jacob et al.~\cite{JacobLF16} for an overview.
  In this paper, we focus on finite automata models and on the notions of opacity that we review in Section~\ref{OpacityDefs}.

  The fastest existing algorithms verifying the notions of opacity under consideration have exponential-time complexity with respect to the number of states of the automaton. In fact, the verification of opacity is a \PSpace{\nobreakdash-}complete problem~\cite{BalunMasopustIFACWC2020,BalunM21,JacobLF16}, and hence we may conclude that there are no polynomial-time algorithms deciding opacity unless $\PTime = \PSpace$.
  Although the assumption that $\PTime \neq \PSpace$ excludes the existence of polynomial-time algorithms, the question whether there is a significantly faster (i.\,e., sub-exponential-time) algorithm remains open.

  To achieve stronger lower bounds (although still conditional), we use the {\em Exponential Time Hypothesis\/} (ETH) and its strong version---the {\em Strong Exponential Time Hypothesis\/} (SETH). Both hypotheses were formulated by Impagliazzo and Paturi~\cite{ImpagliazzoP01} and are based on the observation that (so far) we were not able to find algorithms that would, in the worst case, solve SAT significantly faster than the algorithms trying all possible truth assignments.
  In particular, ETH states that 3-SAT cannot be solved in time $2^{o(n)}$ where $n$ is the number of variables. However, it admits algorithms solving 3-SAT in time $O(c^n)$ where $c<2$. In fact, the current fastest 3-SAT algorithm of Paturi et al.~\cite{PaturiPSZ05}, improved by Hertli~\cite{Hertli14}, runs in time $O^*(1.30704^n)$.
  With increasing $k$, the current fastest $k$-SAT algorithms are getting slower; for instance, the best $4$-SAT algorithm of Hertli~\cite{Hertli14} runs in time $O^*(1.46899^n)$.
  This observation motivated the formulation of SETH that claims that, for any constant $c<2$, there is always a sufficiently large $k$ such that $k$-SAT cannot be solved in time $O(c^n)$~\cite{ImpagliazzoP01}.
  Both hypotheses imply that the complexity classes \PTime and \NP are separated; moreover, SETH implies ETH.

  In this paper, we show that under (S)ETH, there are no significantly faster algorithms verifying opacity. In particular, we show that unless SETH fails, there is no algorithm that decides whether a given $n$-state automaton satisfies the considered notions of opacity and runs in time $O^*(2^{n/c})$, for any $c > 2$ (Theorem~\ref{k-so-seth2} and Corollary~\ref{cor-seth}).
  Since the number of symbols in the alphabet of our construction is unbounded and the standard binary encoding of symbols does not work under SETH, it is not clear whether this result also holds for automata with a fixed size (binary) alphabet.
  We partially explore this question under ETH rather
  than SETH. We show that unless ETH fails, there is no algorithm that decides whether a given $n$-state automaton (over a binary alphabet) satisfies the considered notions of opacity and runs in time $O^*(2^{o(n)})$.
  Our results are summarized in Table~\ref{table01}; for the complexity upper bounds, we refer the reader to the literature~\cite{BalunM21,Saboori2011}.

  As a by-product, we obtain a new conditional lower bound for deciding universality (and hence inclusion and equivalence) for nondeterministic automata (NFA): Unless SETH fails, there is no $c>2$ such that the universality of an $n$-state NFA can be decided in time $O^*(2^{n/c})$ (Corollary~\ref{NFAuniv}). This result strengthens the result of Fernau and Krebs~\cite{FernauK17} showing that if ETH is true, the universality of an $n$-state NFA cannot be decided in time $O^*(2^{o(n)})$.

  \begin{table*}\centering
    \ra{1.1}
    \begin{tabular}{@{}lccl@{}}
      \toprule
      & \multicolumn{2}{c}{Lower bound}  & Upper bound \\
      \cmidrule(lr){2-3}
      & $\Gamma$ not fixed & $|\Gamma|=2$ & \\
      \midrule
          LBO
                    & $O^*\left(2^{n/(2+\eps)}\right)$
                    & $O^*(2^{o(n)})$
                    & \quad$O^*(2^{n})$\\
          CSO
                    & $O^*\left(2^{n/(2+\eps)}\right)$
                    & $O^*(2^{o(n)})$
                    & \quad$O^*(2^{n})$\\
          ISO
                    & $O^*\left(2^{n/(2+\eps)}\right)$
                    & $O^*(2^{o(n)})$
                    & \quad$O^*(2^{n})$\\
          IFO
                    & $O^*\left(2^{n/(2+\eps)}\right)$
                    & $O^*(2^{o(n)})$
                    & $\begin{cases}
                        \mspace{2mu} O^*(2^{n}) & \text{ if } IQ_{NS}=I_{NS}\times F_{NS}\\
                        \mspace{2mu} O^*\bigl(2^{n^2}\bigr) & \text{ otherwise }
                      \end{cases}$
                      \\
          $k$-SO
                    & $O^*\left(2^{n/(2+\eps)}\right)$
                    & $O^*(2^{o(n)})$
                    & \quad$O^*(2^{n})$
                    \\
          INSO
                    & $O^*\left(2^{n/(2+\eps)}\right)$
                    & $O^*(2^{o(n)})$
                    & \quad$O^*(2^{n})$ \\
      \bottomrule
    \end{tabular}
    \medskip
    \caption{An overview of the algorithmic complexity of deciding opacity under the projection $P\colon\Sigma^* \to \Gamma^*$, where $n$ is the number of states of the automaton and $\eps>0$.}
    \label{table01}
  \end{table*}

\section{Preliminaries}
  We assume that the reader is familiar with automata theory~\cite{HopcroftU79}.
  For a set $S$, the cardinality of $S$ is denoted by $|S|$ and the power set of $S$ by $2^{S}$. If $S$ is a singleton, $S=\{x\}$, we often simply write $x$ instead of $\{x\}$. The set of all non-negative integers is denoted by $\mathbb{N}$.

  An alphabet $\Sigma$ is a finite nonempty set of symbols. A string over $\Sigma$ is a finite sequence of symbols from $\Sigma$. The set of all strings over $\Sigma$ is denoted by $\Sigma^*$; the empty string is denoted by $\varepsilon$. A language $L$ over $\Sigma$ is a subset of $\Sigma^*$.
  For a string $u \in \Sigma^*$, the length of $u$ is denoted by $|u|$.
  With every pair of alphabets $(\Sigma,\Gamma)$ with $\Gamma\subseteq \Sigma$, we associate the morphism $P\colon \Sigma^* \to \Gamma^*$ defined by $P(a) = \eps$, for $a\in \Sigma-\Gamma$, and $P(a)=a$, for $a\in \Gamma$; such morphisms are usually called {\em projections}. Intuitively, the action of the projection $P$ is to erase all symbols that do not belong to $\Gamma$; the symbols of $\Gamma$ are usually called {\em observable symbols of $\Sigma$ under the projection $P$}. We lift the projection $P$ from strings to languages in the usual way. The inverse projection of $P$ is the function $P^{-1}\colon \Gamma^* \to 2^{\Sigma^*}$ defined by $P^{-1}(w) = \{w' \in \Sigma^* \mid P(w')=w\}$.

  A {\em nondeterministic finite automaton\/} (NFA) is a structure $\A = (Q,\Sigma,\delta,I,F)$, where $Q$ is a finite  set of states, $\Sigma$ is an input alphabet, $I\subseteq Q$ is a set of initial states, $F \subseteq Q$ is a set of accepting states, and $\delta \colon Q\times\Sigma \to 2^Q$ is a transition function that can be extended to the domain $2^Q\times\Sigma^*$ by induction. If the accepting states are irrelevant, we omit them and simply write $\A=(Q,\Sigma,\delta,I)$.
  The language accepted by $\A$ from the states of $Q_0\subseteq Q$ by the states of $F_0\subseteq F$ is the set $L_m(\A,Q_0,F_0) = \{w\in \Sigma^* \mid \delta(Q_0,w)\cap F_0 \neq\emptyset\}$ and the language generated by $\A$ from the states of $Q_0$ is the set $L(\A,Q_0) = L_m(\A,Q_0,Q)$; in particular, the language accepted by $\A$ is $L_m(\A)=L_m(\A,I,F)$ and the language generated by $\A$ is $L(\A)=L(\A,I)$.
  The NFA $\A$ is {\em deterministic\/} (DFA) if $|I|=1$ and $|\delta(q,a)|\le 1$ for every state $q\in Q$ and every symbol $a \in \Sigma$.

  For an alphabet $\Gamma\subseteq \Sigma$, we define the {\em projected automaton\/} of $\A$, denoted by $P(\A)$, as the reachable part of a DFA obtained from $\A$ by replacing every transition $(q,a,r)$ with the transition $(q,P(a),r)$, followed by the standard subset construction~\cite{HopcroftU79}.

  We define the {\em configuration of $\A$} as the state of the projected automaton $P(\A)$ of $\A$.

  A {\em (Boolean) formula\/} consists of variables, symbols for logical connectives: conjunction, disjunction, negation; and parentheses. A {\em literal\/} is a variable or its negation. A {\em clause\/} is a disjunction of literals. A formula is in {\em conjunctive normal form\/} (CNF) if it is a conjunction of clauses. If each clause has at most $k$ literals, the formula is in $k$-CNF.
  A formula is {\em satisfiable\/} if there is an assignment of 1 and 0 to the variables evaluating the formula to 1. Given a $k\ge 3$ and a formula in $k$-CNF, the $k$-CNF {\em Boolean satisfiability problem\/} ($k$-SAT) is to decide whether the formula is satisfiable.
  If the formula in $k$-CNF has $n$ variables, enumerating all the $2^n$ possible truth assignments results in an $O(2^n n^k)$-time algorithm for $k$-SAT; the polynomial part $O(n^k)$ comes from checking up to $n^k$ clauses.
  We use the notation $O^*$ to hide polynomial factors, that is, $O^*(g(n))=O(g(n)\cdot poly(n))$.

  The exponential time hypothesis states that 3-SAT cannot be solved in sub-exponential time $2^{o(n)}$, where $n$ is the number of variables in the 3-CNF formula~\cite{ImpagliazzoP01}.

  \begin{hypo}[Exponential Time Hypothesis (ETH)]
    There is some $\eps > 0$ such that 3-SAT cannot be solved in time $O(2^{\eps n})$, where $n$ is the number of variables in the formula.
  \end{hypo}

  The strong ETH states that deciding $k$-SAT needs $O^*(2^n)$ time for large $k$~\cite{ImpagliazzoP01}.

  \begin{hypo}[Strong ETH (SETH)]
    For every $\eps > 0$, there is some $k\ge 3$ such that $k$-SAT cannot be solved in time  $O(2^{(1-\eps)n})$.
  \end{hypo}

\section{Opacity Definitions}\label{OpacityDefs}
  We now review the notions of opacity considered in this paper. We distinguish two types of opacity: those representing the secret by strings and those representing the secret by states.

  Language-based opacity is a property asking whether for every secret behaviour, there is a non-secret behaviour that is the same under a considered projection; in this case, an intruder cannot distinguish the secret behaviour from a non-secret behaviour.

  \begin{definition}
    An NFA $\A=(Q,\Sigma,\delta,I)$ is {\em language-based opaque} (LBO) with respect to disjoint languages $L_S,L_{NS} \subseteq L(\A)$, called secret and non-secret languages, respectively, and a projection $P\colon\Sigma^*\to \Gamma^*$ for $\Gamma\subseteq \Sigma$, if $L_S \subseteq P^{-1}P(L_{NS})$.

    The {\sc LBO\/} problem is to decide whether $\A$ is LBO with respect to $L_S$, $L_{NS}$, and $P$.
  \end{definition}

  This definition is general enough to capture other notions, such as {\em strong nondeterministic non-interference} or {\em non-deducibility on composition\/} of Best at al.~\cite{BestDG2010} and Busi and Gorrieri~\cite{BusiG09}, or {\em trace opacity\/} of Bryans et al.~\cite{BryansKMR08}. The secret and non-secret languages are often considered to be regular to ensure that the inclusion problem is decidable~\cite{AsveldN00}.

  State-based opacity hides the secret information into states. In this paper, we consider five notions of state-based opacity.
  Current-state opacity requires that an intruder cannot identify, at any instance of time, whether the system is currently in a secret state.

  \begin{definition}
    An NFA $\A=(Q,\Sigma,\delta,I)$ is {\em current-state opaque} (CSO) with respect to two disjoint sets $Q_S, Q_{NS} \subseteq Q$ of secret and non-secret states, respectively, and a projection $P\colon\Sigma^*\to \Gamma^*$ for $\Gamma\subseteq \Sigma$, if for every string $w\in \Sigma^*$ such that $\delta(I,w)\cap Q_S \neq \emptyset$, there exists a string $w'\in\Sigma^*$ such that $P(w)=P(w')$ and $\delta(I,w')\cap Q_{NS} \neq \emptyset$.

    The {\sc CSO\/} problem is to decide whether $\A$ is CSO with respect to $Q_S$, $Q_{NS}$, and $P$.
  \end{definition}

  Initial-state opacity requires that an intruder can never ascertain whether the computation started in a secret state.

  \begin{definition}
    An NFA $\A=(Q,\Sigma,\delta,I)$ is {\em initial-state opaque} (ISO) with respect to two disjoint sets $I_S,I_{NS}\subseteq I$ of secret and non-secret initial states, respectively, and a projection $P\colon \Sigma^*\rightarrow \Gamma^*$ for $\Gamma\subseteq \Sigma$, if for every $w \in L(\A,I_S)$, there exists $w' \in L(\A,I_{NS})$ such that $P(w) = P(w')$.

    The {\sc ISO\/} problem is to decide whether $\A$ is ISO with respect to $I_S$, $I_{NS}$, and $P$.
  \end{definition}

  Initial-and-final-state opacity~\cite{WuLafortune2013} generalizes both CSO and ISO. The secret is represented as a set of pairs of an initial state and of an accepting state. Therefore, ISO is a special case of initial-and-final-state opacity where the accepting states do not play a role, while CSO is a special case where the initial states do not play a role.

  \begin{definition}
    An NFA $\A=(Q,\Sigma,\delta,I,F)$ is {\em initial-and-final-state opaque} (IFO) with respect to two disjoint sets $IQ_S,IQ_{NS}\subseteq I \times F$ of secret and non-secret pairs of states, respectively, and a projection $P\colon \Sigma^*\rightarrow \Gamma^*$ for $\Gamma\subseteq \Sigma$, if for every secret pair $(q_0,q_f) \in IQ_S$ and every string $w \in L_m(\A,q_0,q_f)$, there exists a non-secret pair $(q_0',q_f') \in IQ_{NS}$ and a string $w' \in L_m(\A,q_0',q_f')$ such that $P(w) = P(w')$.

    The {\sc IFO\/} problem is to decide whether $\A$ is IFO with respect to $IQ_S$, $IQ_{NS}$, and $P$.
  \end{definition}

  The algorithmic time complexity of deciding {\sc IFO\/} is known to be $O^*(2^{n^2})$ in general, and $O(2^{2n})$ if $IQ_S=I_S\times F_S$ and $IQ_{NS}=I_{NS}\times F_{NS}$, for some $I_S,I_{NS}\subseteq I$ and $F_S,F_{NS}\subseteq F$~\cite{WuLafortune2013}. Our complexity in Table~\ref{table01} is based on the following observations.
  
  Consider an NFA $\A=(Q,\Sigma,\delta,I,F)$ and two sets $IQ_S, IQ_{NS}\subseteq I\times F$. The IFO property of $\A$ is unchanged if all pairs $(s,f_1),(s,f_2),\ldots,(s,f_k)$ with a common left component are replaced by a single pair of the form $(s,\{f_1,f_2,\ldots,f_k\})$. This reduces the number of pairs to be considered to $n$, where $n$ is the number of states of $\A$. For every pair $(s_i,F_i)\in I\times 2^F$, we define the language
  $
    L_{i} = L_m(\A,s_i,F_i)
  $
  and the languages
  \[
    L_S = \bigcup_{(s_i,F_i)\in IQ_S} L_i \quad\text{ and }\quad L_{NS} = \bigcup_{(s_i,F_i)\in IQ_{NS}} L_i\,.
  \]

  Then, deciding whether $\A$ is IFO with respect to $IQ_S$, $IQ_{NS}$, and $P$ is equivalent to deciding whether the inclusion $P(L_S)\subseteq P(L_{NS})$ holds true. Since both $L_S$ and $L_{NS}$ can be represented by NFAs consisting of at most $n$ copies of $\A$, they have $O(n^2)$ states. The inclusion $P(L_S)\subseteq P(L_{NS})$ of languages of two NFAs can be tested in time $O(n^2 2^{n^2}) = O^*(2^{n^2})$, which is a complexity upper bound that coincides with the bound of Saboori and Hadjicostis~\cite{SabooriH2013}, who used trellis automata.

  If $IQ_{NS}= I_{NS} \times F_{NS} \subseteq I\times F$, then the NFA for $L_{NS}$ coincides with $\A$ where the initial states are $I_{NS}$ and the final states are $F_{NS}$. In particular, this automaton has $n$ states, and therefore the inclusion $P(L_S)\subseteq P(L_{NS})$ can be tested in time $O(n^2 2^n)=O^*(2^n)$.

  \medskip
  The notion of $k$-step opacity generalizes CSO by requiring that the intruder cannot ascertain the secret in the current and $k$ subsequent states. By definition, CSO is equivalent to 0-step opacity. We use a slight generalisation of a definition of Saboori and Hadjicostis~\cite{SabooriH12a} that was formulated by~Balun and Masopust~\cite{BalunM21}.

  \begin{definition}
    An NFA $\A=(Q,\Sigma,\delta,I)$ is {\em $k$-step opaque} ($k$-SO), for a given $k\in\mathbb{N}\cup\{\infty\}$, with respect to two disjoint sets $Q_S,Q_{NS}\subseteq Q$ of secret and non-secret states, respectively, and a projection $P \colon \Sigma^* \to \Gamma^*$ for $\Gamma\subseteq \Sigma$, if for every string $st \in L(\A)$ such that $|P(t)| \leq k$ and $\delta(\delta(I, s)\cap Q_S, t) \neq \emptyset$, there exists a string $s't' \in L(\A)$ such that $P(s)=P(s')$, $P(t)=P(t')$, and $\delta (\delta(I,s')\cap Q_{NS}, t') \neq \emptyset$.

    The {\sc $k$-SO\/} problem is to decide whether $\A$ is $k$-SO with respect to $Q_S$, $Q_{NS}$, and $P$.
  \end{definition}

  A special case of $k$-SO for $k$ being infinity is called infinite-step opacity (INSO). These two notions are closely related for finite automata, because an $n$-state automaton is infinite-step opaque if and only if it is $(2^n-2)$-step opaque~\cite{Yin2017}.

  For the state-based opacity notions, we may assume, without loss of generality, that the projection $P$ is an identity; indeed, the transitions labeled by symbols from $\Sigma-\Gamma$ can be seen as $\eps$-transitions, and hence they can be removed by the classical algorithm eliminating $\eps$-transitions~\cite{HopcroftU79}. This algorithm does not change the number of states but can quadratically increase the number of transitions. In the sequel, we omit the projection if it is an identity.

\section{Lower Bounds under Strong Exponential Time Hypothesis}
  We now show that under the strong exponential time hypothesis, there is no algorithm deciding current-state opacity that would be significantly faster than the best known algorithm.

  \begin{theorem}\label{k-so-seth2}
    Unless SETH fails, there is no algorithm deciding whether a given $n$-state NFA is {\sc CSO\/} that runs in time $O^*(2^{n/(2+\eps)})$, for any $\eps > 0$.
  \end{theorem}
  \begin{proof}
    For a given formula $\varphi$ in $k$-CNF with $n$ variables $X=\{x_1,\ldots,x_n\}$ and $m$ clauses $C=\{c_1,\ldots,c_m\}$, we construct, in polynomial time, an instance of {\sc CSO\/} consisting of an NFA $\A_\varphi$ with $N=2n+2$ states and of sets of secret and non-secret states $Q_S$ and $Q_{NS}$, respectively, such that $\A_\varphi$ is CSO with respect to $Q_S$ and $Q_{NS}$ if and only if $\varphi$ is satisfiable.
    As a result, if there was an algorithm solving {\sc CSO\/} in time $O^*(2^{N/(2+\eps)})$, then there would be an algorithm solving $k$-SAT in time $O({\rm poly}(n))+O^*(2^{(2n+2)/(2+\eps)})=O^*(2^{(1-\delta)n})$, for $\delta = \eps/(2+\eps) > 0$, which contradicts SETH, and proves the theorem.

    Intuitively, we construct the NFA $\A_\varphi
    $ such that when $A_{\varphi}$ reads a string of a particular type (based on Zimin words), it is forced to examine all possible assignments to the variables $x_1,\dots,x_n$. If none of the assignments satisfies $\varphi$, then, after reading the whole string, the automaton $A_{\varphi}$ ends up in a configuration that contains only secret states, rendering thus $A_{\varphi}$ not CSO. On the other hand, when a satisfying assignment is encountered (or the string is not of the particular type), a non-secret state is permanently added to the configuration of $A_{\varphi}$, and hence $A_{\varphi}$ is CSO.

    \begin{figure}
      \centering
      \includegraphics[scale=.92]{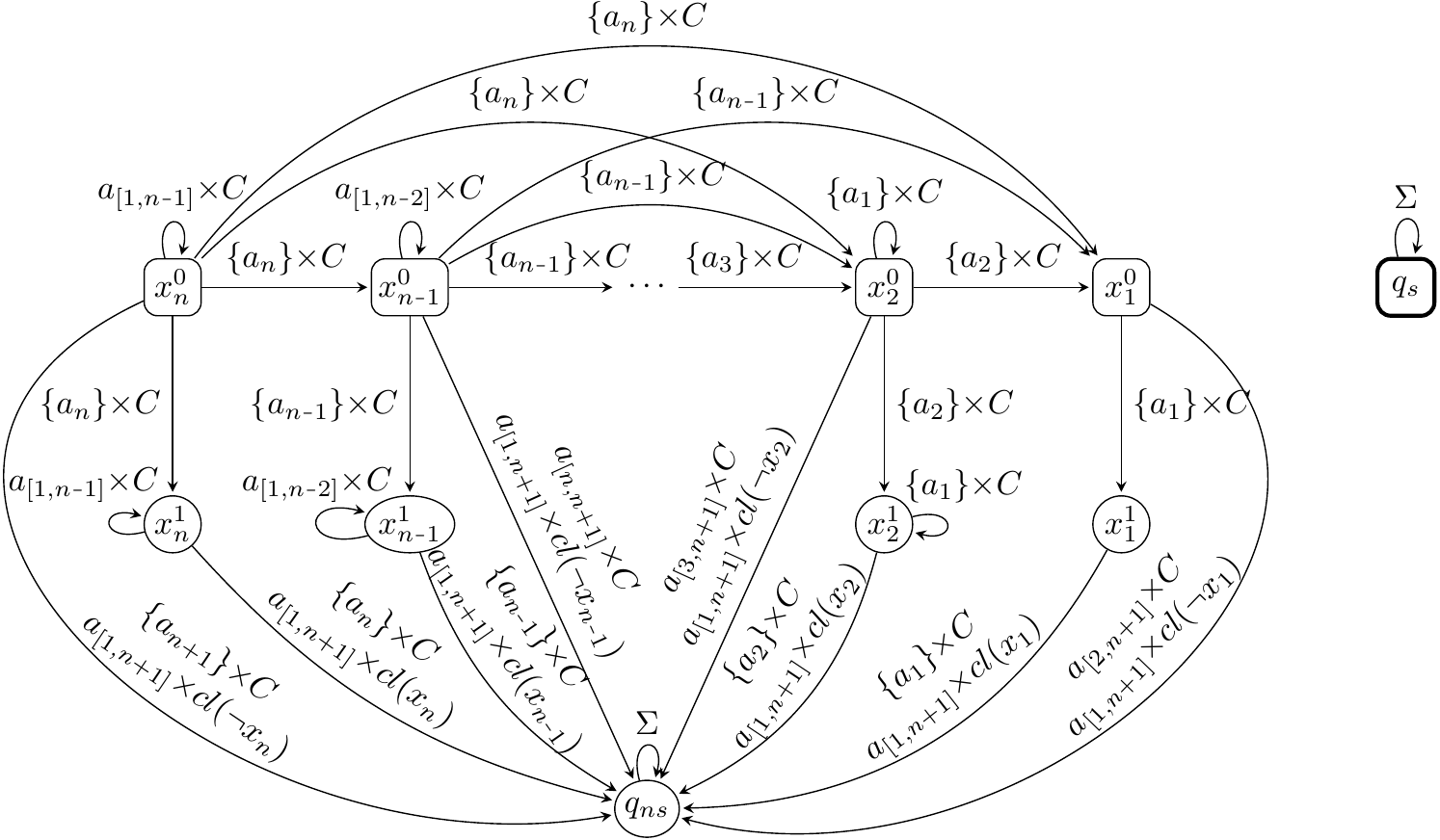}
      \caption{The NFA $A_{\varphi}$ of Theorem~\ref{k-so-seth2}, where the initial states are squared, the single secret state $q_s$ is in bold, and for positive integers $i\leq j$, $[i,j]=\{i,i+1,\ldots,j\}$ and $a_{[i,j]}=\{a_r\mid r\in [i,j]\}$.}
      \label{fig:kso-seth-cons}
    \end{figure}

    Formally, the NFA $\A_{\varphi}=(Q,\Sigma,\delta,I)$, where
      the set of states is $Q=\{q_s,q_{ns}\}\cup \{x_i^0, x_i^1 \mid x_i \in X\}$ with $x_i^r$ representing the assignment of $r\in\{0,1\}$ to the variable $x_i$,
      the alphabet $\Sigma = \Gamma = \{a_1,\ldots, a_{n+1}\} \times C$, that is, the projection $P$ is an identity, and
      the set of secret states is $Q_{S}=\{q_{s}\}$, that is, the state $q_{s}$ is the only secret state, the remaining states are non-secret.
    For an illustration of the construction, the reader may follow Example~\ref{ex01} together with the rest of the proof.

    Let $L$ be the set of literals of $\varphi$. We use the function $cl\colon L \rightarrow 2^C$ that assigns to a literal $\ell$ the set $cl(\ell) = \{c \in C \mid \ell\in c\}$ of clauses containing $\ell$, and define the transition function $\delta$ as follows, see Figure~\ref{fig:kso-seth-cons} for an illustration:
    \begin{outline}
      \1 The self-loops $(q_s,\sigma,q_s)$ and $(q_{ns},\sigma,q_{ns})$ belong to $\delta$ for every $\sigma\in\Sigma$;
      \1 For every state $x_i^0$ and every $c\in C$,
        \2 the transition $(x_i^0, (a_i,c), x_i^1)\in\delta$;
        \2 the self-loop $(x_i^0, (a_j,c), x_i^0) \in \delta$ for $1\le j\le i-1$;
        \2 the transition $(x_i^0, (a_i,c), x_j^0) \in \delta$ for $1\le j\le i-1$;
        \2 the transition $(x_i^0, (a_j,c), q_{ns})\in\delta$ for $i+1\le j\le n+1$;
        \2 the transition $(x_i^0, (a_j,c),q_{ns})\in\delta$ for $1\le j\le n+1$ and $c\in cl(\neg x_i)$;

      \1 For every state $x_i^1$ and $c\in C$,
        \2 the transition $(x_i^1, (a_i,c), q_{ns})\in\delta$;
        \2 the self-loop $(x_i^1, (a_j,c), x_i^1)\in\delta$, for $1\le j\le i-1$;
        \2 the transition $(x_i^1, (a_j,c), q_{ns})\in\delta$ for $1\le j\le n+1$ and $c\in cl(x_i)$.
    \end{outline}
    Finally, the set of initial states is
    $
      I = \{q_s\} \cup \{x_1^0, \ldots, x_n^0\}
    $,
    corresponding to the assignment of 0 to all variables of $\varphi$.

    We now define a language $W_\varphi = Z_n \cdot (\{a_{n+1}\}\times C)$, where $Z_n$ is a language over the alphabet $\{a_1,a_2,\ldots,a_{n}\}\times C$ recursively defined as follows:
    \[
      Z_1=\{a_1\}\times C \qquad \text{ and }\qquad
      Z_i = Z_{i-1} \cdot (\{a_{i}\}\times C) \cdot Z_{i-1}, \quad\text{for $1<i\le n$}\,.
    \]
    Such strings are known as {\em Zimin words\/} and it is well known that any string of $Z_n$ is of length $2^n-1$~\cite{sloan2} and that
    \begin{equation}\label{key}
      \parbox{0.9\textwidth}{the symbol on the $\ell$th position of any string from $Z_n$ is of the form $\{a_j\}\times C$, where $j-1$ is the number of trailing zeros in the binary representation of $\ell$ \cite{sloan}.}
    \end{equation}

    We finish the proof in a series of claims. The first claim shows that along any string of $Z_n$, the states $\{x_i^0, x_i^1 \mid x_i \in X\}$ of $\A_\varphi$ encode all possible assignments to the variables.

    \begin{restatable}{claim}{claimA}\label{claim3a}
      Let $\A_\varphi^X$ denote $\A_\varphi$ without the states $q_s$ and $q_{ns}$ and the corresponding transitions. For every $w \in Z_n$, after reading the prefix of $w$ of length $\ell \le 2^n-1$, the configuration of $\A_\varphi^X$ is $\{x_{n}^{r_{n}}, x_{n-1}^{r_{n-1}}, \ldots, x_1^{r_1}\}$, where $r_{n}r_{n-1}\cdots r_1$ represents $\ell$ in binary.
    \end{restatable}

    \begin{restatable}{claim}{claimB}\label{claim1}
      Every configuration of $\A_\varphi$ contains the secret state $q_s$.
    \end{restatable}

    By Claims~\ref{claim3a} and~\ref{claim1}, and because only the state $q_s$ itself is reachable from $q_s$ and only the state $q_{ns}$ itself is reachable from $q_{ns}$, we have the following observation specifying the computation of $\A_\varphi$ along the strings of $Z_n$.

    \begin{claim}\label{claim3b}
      After reading the prefix of $w\in Z_n$ of length $\ell \le 2^n-1$, the configuration of $\A_\varphi$ is either $\{x_{n}^{r_{n}}, x_{n-1}^{r_{n-1}}, \ldots, x_1^{r_1}\}\cup\{q_s\}$ or $\{x_{n}^{r_{n}}, x_{n-1}^{r_{n-1}}, \ldots, x_1^{r_1}\}\cup\{q_s,q_{ns}\}$, where $r_{n}r_{n-1}\cdots r_1$ represents $\ell$ in binary.
    \claimqedhere{}\end{claim}

    If there is a satisfying assignment, then the non-secret state $q_{ns}$ is reached by $\A_\varphi$.

    \begin{restatable}{claim}{claimC}\label{claim3c}
      For every prefix $w$ of a string in $Z_n$, if the configuration of $\A_\varphi$ after reading $w$ satisfies $\varphi$, then after reading any further symbol in $\Sigma$, the configuration of $\A_\varphi$ contains $q_{ns}$.
    \end{restatable}

    We now show that if $\varphi$ is satisfiable, then $\A_\varphi$ is CSO. To this end, we consider an arbitrary string $w$ over $\Sigma$, and we denote by $u$ the longest prefix of $w$ that is a prefix of a string in $Z_n$. Let $\ell$ be the minimal number such that its binary representation $r_{n}r_{n-1}\cdots r_1$ is a satisfying assignment to the variables of $\varphi$.
    
    If $\ell \leq |u|$, then, by Claim~\ref{claim3b}, the configuration of $\A_\varphi$ after reading any prefix of $u$ of length $\ell' \le \ell$ contains non-secret states $x_{n}^{r_{n}}, x_{n-1}^{r_{n-1}}, \ldots, x_1^{r_1}$, where $r_{n}r_{n-1}\cdots r_1$ represents $\ell'$ in binary, and, by Claim~\ref{claim3c}, the ($\ell+1$)st symbol of $w$ moves $\A_\varphi$ to a configuration that contains $q_{ns}$; that is, $\A_\varphi$ is CSO.

    If $\ell > |u|$, we have $w=u(a_s,c)v$ for $(a_s,c) \in \Sigma$ and $v \in \Sigma^*$. Because $\varphi$ is satisfiable, we have $|u| < 2^n - 1$. By Claim~\ref{claim3b}, the configuration of $\A_\varphi$ after reading $u$ contains $x_n^{r_n},\ldots,x_1^{r_1}$, where $r_{n}r_{n-1}\cdots r_1$ represents $|u|$ in binary. Let $r_t$ be the rightmost zero of $r_{n}r_{n-1}\cdots r_1$, that is, $r_{n}r_{n-1}\cdots r_1 = r_{n}r_{n-1}\cdots r_{t+1} 0 1 \cdots 1$. Then, $|u|+1$ is $r_{n}r_{n-1}\cdots r_{t+1} 1 0 \cdots 0$ in binary and, by~\eqref{key}, the symbol $(a_s,c) \notin \{a_{t}\}\times C$.  However, if $s<t$, then $x_s^1$ goes to state $q_{ns}$ under $\{a_s\}\times C$, while if $s>t$, then $x_t^0$ goes to state $q_{ns}$ under $\{a_s\}\times C$. In both cases, the non-secret state $q_{ns}$ is in the next configuration of $\A_\varphi$, and hence $\A_\varphi$ is CSO.
    
    \medskip
    To prove that if $\varphi$ is not satisfiable, then $\A_\varphi$ is not CSO, we use the following claim.

    \begin{restatable}{claim}{claimD}\label{claim3d}
      If $\varphi$ is not satisfiable, there is a string $w_\varphi \in Z_n$ such that the configuration of $\A_\varphi$ after reading $w_\varphi$ is $\{x_n^1,x_{n-1}^1,\dots,x_1^1\} \cup \{q_s\}$.
    \end{restatable}

    We now show that if $\varphi$ is not satisfiable, then $\A_\varphi$ is not CSO. To this end, we consider the string $w_\varphi$ constructed in Claim~\ref{claim3d}, which we extend to a string from $W_\varphi = Z_n \cdot (\{a_{n+1}\}\times C)$ by adding a symbol of the form $\{a_{n+1}\}\times C$.
    Since $\varphi$ is not satisfiable, there is a clause $c\in C$ that is not satisfied by the assignment of 1 to the variables; that is, there is $c \notin \bigcup_{i=1}^{n} cl(x_i)$.
    Then, the string $w_\varphi (a_{n+1},c)$ moves the automaton $\A_\varphi$ from the configuration $\{x_{n}^{1}, x_{n-1}^{1}, \ldots, x_1^{1}\}\cup \{q_s\}$ to the configuration $\{q_s\}$, and hence $\A_\varphi$ is not CSO.
  \qed\end{proof}

  \begin{figure}
    \centering
    \includegraphics[scale=1]{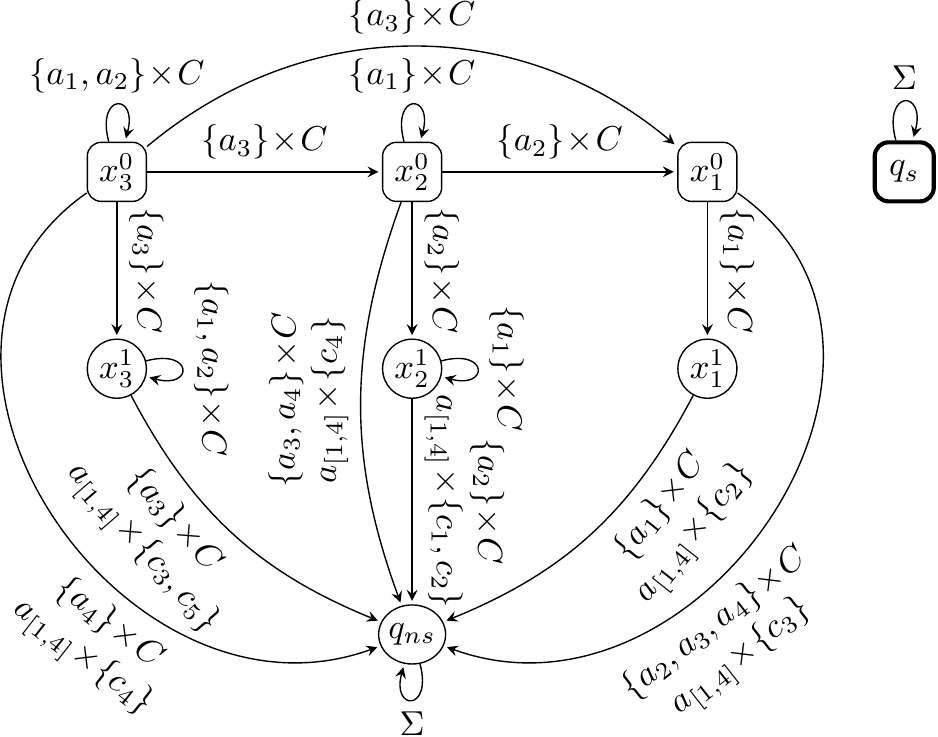}
    \caption{The NFA $A_{\varphi}$ illustrating Theorem~\ref{k-so-seth2}; the initial states are squared.}
    \label{fig:ksat-cso-example-2}
  \end{figure}

  We now illustrate the construction.
  \begin{example}\label{ex01}
    For simplicity, we consider a 2-CNF formula
    \[
      \varphi = (x_2\lor x_2)\wedge (x_1\lor x_2)\wedge (\neg x_1\lor x_3)\wedge (\neg x_2\lor \neg x_3)\wedge (x_3)
    \]
    with three variables $x_1$, $x_2$, $x_3$ and five clauses $c_1=\{x_2,x_2\}$, $c_2=\{x_1, x_2\}$, $c_3=\{\neg x_1,x_3\}$, $c_4=\{\neg x_2, \neg x_3\}$, and $c_5=\{x_3\}$. The automaton $\A_\varphi=(Q,\Sigma, \delta, \{q_s, x_1^0, x_2^0, x_3^0\})$ is depicted in Figure~\ref{fig:ksat-cso-example-2}, where $Q=\{q_{s},q_{ns}\}\cup \{x_1^0, x_1^1, x_2^0, x_2^1, x_3^0, x_3^1\}$, $\Sigma = \Gamma = \{a_1,a_2,a_3,a_4\}\times \{c_1,c_2,c_3,c_4,c_5\}$, and $q_s$ is the only secret state.
    Since $\varphi$ is not satisfiable, $\A_\varphi$ is not CSO; indeed, the string $w = (a_1,c_1)(a_2,c_2)(a_1,c_5)(a_3,c_3)(a_1,c_1)(a_2,c_1)(a_1,c_4)(a_4,c_4)$ moves $\A_\varphi$ to the configuration $\{q_{s}\}$ consisting solely of the secret state, cf.~Figure~\ref{fig:ksat-cso-example-obs2-a} depicting the reachable configurations of $\A_\varphi$.

    On the other hand, if we consider the formula $\varphi' = c_1\wedge c_2\wedge c_3 \wedge c_4$, then $\varphi'$ is satisfiable, and hence the NFA $\A_{\varphi'}$ obtained from $\A_\varphi$ by removing all transitions under symbols containing $c_5$, is CSO; it is visible from the reachable configurations of $\A_{\varphi'}$ depicted in~Figure~\ref{fig:ksat-cso-example-obs2}.
  \end{example}

  \begin{figure}
    \centering
    \includegraphics[scale=1]{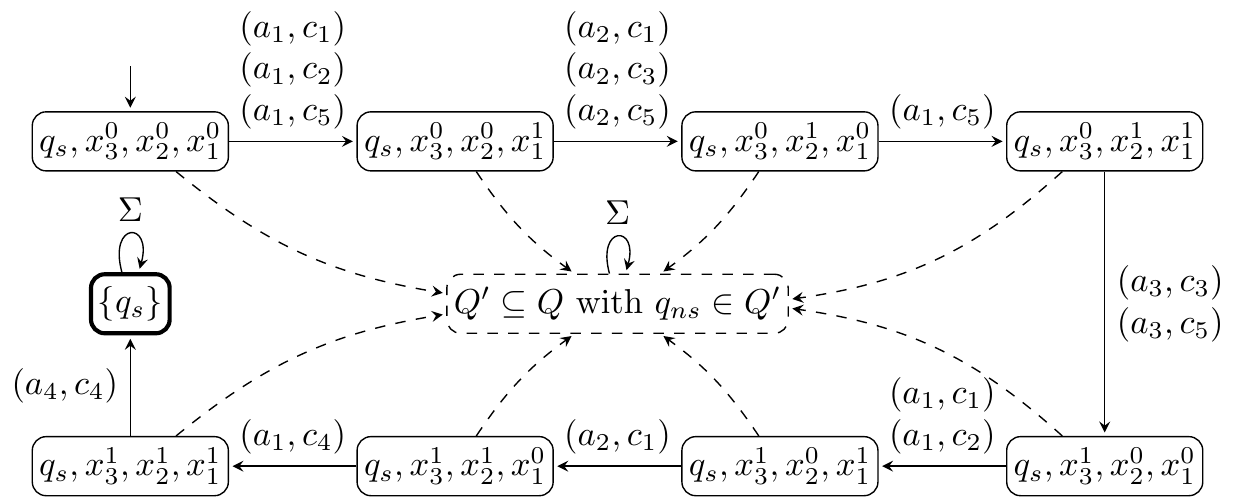}
    \caption{The configurations of $A_{\varphi}$---all undefined transitions go to the dashed middle state.}
    \label{fig:ksat-cso-example-obs2-a}
  \end{figure}

  \begin{figure}
    \centering
    \includegraphics[scale=1]{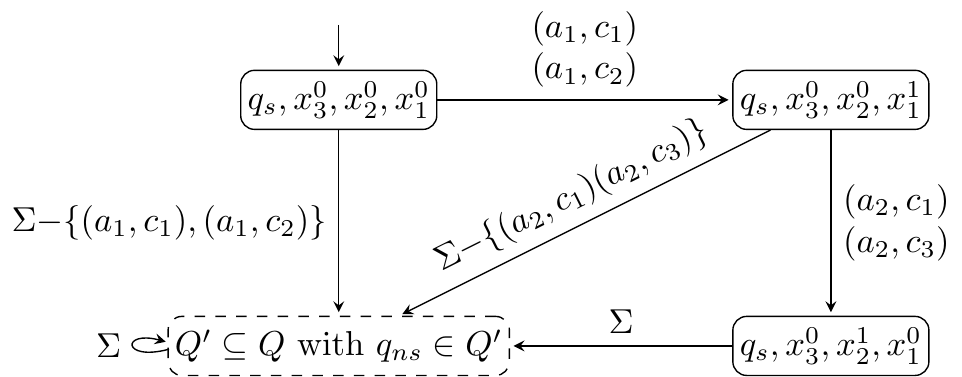}
    \caption{The configurations of $A_{\varphi'}$.}
    \label{fig:ksat-cso-example-obs2}
  \end{figure}

  The considered problems are all \PSpace-complete, and hence reducible to each other in polynomial time. However, this fact does not provide us with much information about the reductions. Even though some particular reductions have been discussed in the literature by Wu and Lafortune~\cite{WuLafortune2013} and Balun and Masopust~\cite{BalunM21,BalunM22}, they are in most cases not suitable to prove lower bounds.

  We now discuss the case of other types of opacity.

  \begin{corollary}\label{cor-seth}
    Unless SETH fails, there is no algorithm deciding if a given $n$-state NFA is LBO/ISO/IFO/$k$-SO/INSO that runs in time $O^*(2^{n/(2+\eps)})$, for any $\eps > 0$.
  \end{corollary}
  \begin{proof}
    Consider the instance of CSO given by the NFA $\A$ and the sets $Q_{S}$ and $Q_{NS}$ constructed in the proof of Theorem~\ref{k-so-seth2}. Then, $\A$ is {\sc CSO\/} with respect to $Q_S=\{q_s\}$ and $Q_{NS}=Q-Q_S$ if and only if $\A$ is LBO with respect to $L_S=L_m(\A,q_s,q_s)=\Sigma^*$ and $L_{NS}=L_m(\A,\{x_1^0, \ldots, x_n^0\},Q-\{q_s\})$. Since the parts of $\A$ corresponding to languages $L_S$ and $L_{NS}$ are disjoint, the instance of {\sc LBO\/} may be encoded directly into $\A$ by defining the corresponding states accepting the languages
    $L_S$ and $L_{NS}$. Hence, the instance of {\em LBO\/} is of the same size as the instance of {\sc CSO}. Therefore, if we solved the instance of {\sc LBO\/} in time $O^*(2^{(1-\delta)n})$, we would also solve the instance of {\sc CSO\/} in time $O^*(2^{(1-\delta)n})$.

    Since there is no transition in $\A$ from the sole secret state $q_s$ to another state, the NFA $\A$ is CSO if and only if $\A$ is $k$-SO, for any $k\in\mathbb{N}\cup\{\infty\}$, and hence the result holds for {\sc $k$-SO\/} as well as for {\sc INSO}.
    
    Furthermore, the NFA $\A$ is {\sc CSO\/} with respect to the sets $Q_{S}$ and $Q_{NS}$ if and only if $\A$ is ISO with respect to the secret initial state $I_S=\{q_s\}$ and non-secret initial states $I_{NS}=\{x_1^0, \ldots, x_n^0\}$. Indeed, since $L(\A,I_S)=\Sigma^*$, the NFA $\A$ is not CSO if and only if there is a string $w\in\Sigma^*$ that moves $\A$ from the initial configuration $I_{NS}$ to the configuration $\emptyset$, which is if and only if $\A$ is not ISO. As a result, solving {\sc ISO\/} in time $O^*(2^{(1-\delta)n})$ would solve {\sc CSO} in time $O^*(2^{(1-\delta)n})$.

    Finally, if all states of $\A$ are accepting, then $\A$ is ISO with respect to $I_S=\{q_s\}$ and $I_{NS}=\{x_1^0, \ldots, x_n^0\}$ if and only if $\A$ is IFO with respect to $IQ_S=\{(q_s,q_s)\}$ and $IQ_{NS}=I_{NS}\times Q$; hence, solving {\sc IFO} in time $O^*(2^{(1-\delta)n})$ would solve {\sc ISO} in time $O^*(2^{(1-\delta)n})$.
  \end{proof}

\section{Lower Bounds under Exponential Time Hypothesis}
  The number of symbols in the NFA constructed in Theorem~\ref{k-so-seth2} depends on the number of clauses in the instance of SAT. Since the standard binary encoding of symbols does not work under SETH, it is an open problem whether the results of Theorem~\ref{k-so-seth2} and Corollary~\ref{cor-seth} also hold for a fixed-sized alphabet.

  Although we do not answer this question, we provide a lower bound for NFAs over a binary alphabet under ETH. Namely, we show that there is no algorithm solving the considered notions of opacity for such $n$-state NFAs that runs in time $O^*(2^{o(n)})$.
  We obtain the result by adjusting the construction of Fernau and Krebs~\cite{FernauK17}, who showed that there is no algorithm solving the universality problem for $n$-state NFAs over a binary alphabet that runs in time $O^*(2^{o(n)})$ unless ETH fails, and by using the observation of Cassez et al.~\cite{Cassez2012} that universality can be reduced to opacity.

  \begin{theorem}\label{k-so-eth}
    Unless ETH fails, there is no algorithm deciding whether a given $n$-state NFA (over a binary alphabet) is CSO that runs in time $O^*(2^{o(n)})$.
  \end{theorem}
  \begin{proof}
    A 3-coloring of a graph $G=(V,E)$ is a function $\mu\colon V \to \{a,b,c\}$. The coloring is proper if $\mu(u) \neq \mu(v)$ whenever $uv \in E$. The {\sc 3-Coloring} problem is to decide, given a graph $G$, whether there is a proper 3-coloring of $G$.

    For a graph $G$ with $n$ vertices, $V=\{v_1,v_2,\ldots,v_n\}$, and $m$ edges, we construct an NFA $\A=(Q,\Sigma,\delta,I)$, where the states are $Q = \{s,f\} \cup \{q_1,\ldots,q_n\} \cup \{x_1,\ldots,x_{n-1} \mid x\in\{a, b, c\} \}$, the alphabet is $\Sigma = \Gamma = \{a,b,c\}$, the initial state is $I=\{q_1\}$, the secret state is $Q_S=\{s\}$, and the non-secret states are $Q_{NS}=Q-Q_S$.
    We define the transition function $\delta$ as shown in Figure~\ref{reduction01}, and further extended it by adding three transitions $(q_i,a,{a}_{j-i})$, $(q_i,b,{b}_{j-i})$, and $(q_i,c,{c}_{j-i})$ for every edge $v_iv_j\in E$ with $i<j$.

    \begin{figure}
      \centering
      \includegraphics[scale=1]{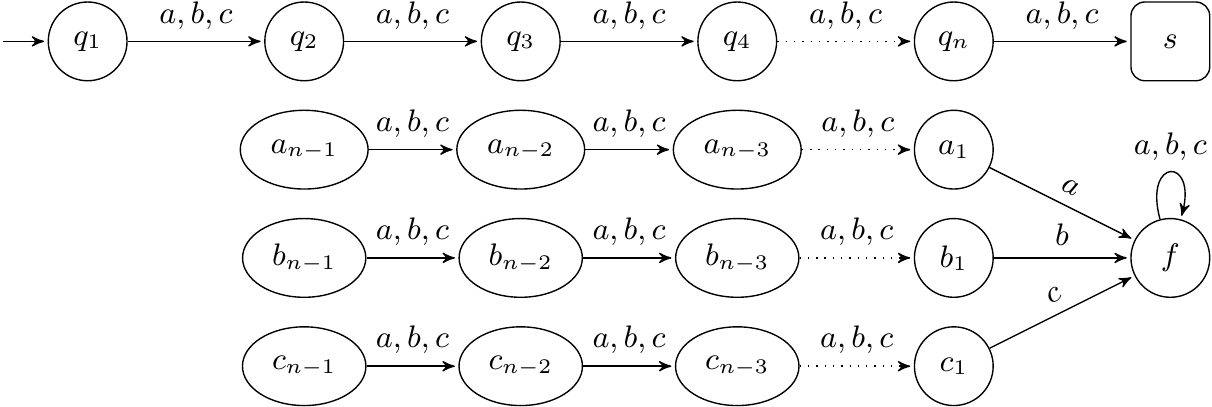}
        \caption{The main part of the NFA $\A$ resulting from the reduction of Theorem~\ref{k-so-eth} without the transitions corresponding to the edges of $G$. The secret state is state $s$.}
        \label{reduction01}
      \end{figure}

    Intuitively, the coloring of $G$ is encoded as a string $w=c_1\cdots c_n$ of length $n$, where $c_i$ is the color of vertex $i$, and $\A$ is CSO with respect to $\{s\}$ and $Q-\{s\}$ if and only if the non-secret state $f$ is reachable under $w$; indeed, the secret state $s$ is reachable under every string of length $n$.

    Fernau and Krebs~\cite{FernauK17} showed that $\A$ ends up in state $f$ under $w$ if and only if $w$ encodes a coloring that is not proper.
    Therefore, if $w$ is a proper 3-coloring of $G$, then $w$ does not move $\A$ to state $f$; that is, only the secret state $s$ is reached under $w$, and hence $\A$ is not CSO with respect to $\{s\}$ and $Q-\{s\}$.
    On the other hand, if $G$ does not have a proper 3-coloring, then every string of length $n$ moves $\A$ to both secret state $s$ and non-secret state $f$; that is, $\A$ is CSO with respect to $\{s\}$ and $Q-\{s\}$.

    If $G$ has $n$ vertices and $m$ edges, then $\A$ has $N=4n-1$ states and $M=12n+3m-12$ transitions. If there was an $O^*(2^{o(N)})$-time algorithm deciding {\sc CSO}, we could reduce the instance of {\sc 3-Coloring} to an instance of {\sc CSO} in time $O(N+M)$, and solve {\sc CSO} in time $O^*(2^{o(N)})$. Altogether, we could solve {\sc 3-Coloring} in time $O(N+M) + O^*(2^{o(N)}) = O^*(2^{o(n)})$, which contradicts ETH.
  \end{proof}

  Similarly to the discussion in the previous section, to prove the lower bound for the other notions of opacity, it seems natural to combine the construction of the previous proof with the existing reductions among the notions~\cite{BalunM21,BalunM22,WuLafortune2013}. However, most of the reductions result in too large, though polynomial, instances, and hence they are not suitable for our purposes. Therefore, new reductions are needed.

 \begin{corollary}\label{cor-k-so-eth}
    Unless ETH fails, there is no algorithm deciding if a given $n$-state NFA (over a binary alphabet) is LBO/ISO/IFO/$k$-SO/INSO that runs in time $O^*(2^{o(n)})$.
  \end{corollary}
  \begin{proof}
    For the NFA $\A$ of Theorem~\ref{k-so-eth}, we have $\A$ is CSO with respect to $\{s\}$ and $Q-\{s\}$ if and only if $\A$ is LBO with respect to $L_S=L_m(\A,q_1,s)$ and $L_{NS}=L_m(\A,q_1,Q-\{s\})$~\cite{WuLafortune2013}. If we could solve {\sc LBO} in time $O^*(2^{o(n)})$, we could solve {\sc 3-Coloring\/} in time $O^*(2^{o(n)})$.

    Furthermore, since there is no transition from the sole secret state $s$, the automaton $\A$ is CSO with respect to $\{s\}$ and $Q-\{s\}$ if and only if $\A$ is $k$-SO with respect to $\{s\}$ and $Q-\{s\}$, for any $k\in\mathbb{N}\cup\{\infty\}$. Therefore, the result holds for {\sc $k$-SO\/} as well as for {\sc INSO}.

    Now, we take the NFA $\A$ and add a copy of states $q_1,q_2,\ldots,q_n$, denoted by $q_1', q_2',\ldots,q_n'$, together with all transitions to states different from $s$, that is, we add $(q_i',x,p)$ for every transition $(q_i,x,p)$ with $p\neq s$. We set the states $q_1$ and $q_1'$ initial, and denote the result by $\A'$. Then, the NFA $\A$ is CSO with respect to $\{s\}$ and $Q-\{s\}$ if and only if $\A'$ is ISO with respect to $I_S=\{q_1\}$ and $I_{NS}=\{q_1'\}$.
    Indeed, if $\A$ is CSO with respect to $\{s\}$ and $Q-\{s\}$, then for every $w$ moving $\A$ to state $s$, there is $w'$ moving $\A$ to state $f$; and so do the strings $w$ and $w'$ in $\A'$, which shows that $\A'$ is ISO with respect to $I_S=\{q_1\}$ and $I_{NS}=\{q_1'\}$.
    On the other hand, if $\A$ is not CSO with respect to $\{s\}$ and $Q-\{s\}$, then there is $w$ moving $\A$ only to state $s$, and hence $w$ cannot be read by $\A'$ from state $q_1'$, which shows that $\A'$ is not ISO with respect to $I_S=\{q_1\}$ and $I_{NS}=\{q_1'\}$.
    If we could solve {\sc ISO\/} in time $O^*(2^{o(n)})$, we could solve {\sc 3-Coloring\/} in time $O^*(2^{o(n)})$ by reducing it to {\sc ISO\/} and solving {\sc ISO\/} in time $O^*(2^{o(N+n)})=O^*(2^{o(n)})$, for $N=4n-1$.

    If we in addition set the states $s$ and $f$ accepting, then $\A'$ is ISO with respect to $I_S=\{q_1\}$ and $I_{NS}=\{q_1'\}$ if and only if $\A'$ is IFO with respect to $IQ_S=\{(q_1,s)\}$ and $IQ_{NS}=\{(q_1',f)\}$, and hence if we could solve {\sc IFO\/} in time $O^*(2^{o(n)})$, we could solve {\sc 3-Coloring\/} in time $O^*(2^{o(n)})$.
  \end{proof}

\section{Discussion and Conclusions}
  We showed that if the strong exponential time hypothesis holds true, then, for any $c>2$, there are no algorithms deciding various types of opacity in time $O^*(2^{n/c})$. Therefore, the current algorithms cannot be significantly improved.

  More precisely, the results say that there are no algorithms deciding various types of opacity in time $O^*(\sqrt{2}^n) = O^*(1.414213562^n)$. However, the results admit the existence of algorithms deciding opacity in time $O^*(1.5^n)$. Whether such algorithms exist or whether the current lower bounds can be strengthen remains an open problem.

  The construction used in the proof of Theorem~\ref{k-so-seth2} can be utilized to improve the conditional lower bound of deciding universality for NFAs. The universality problem for NFAs asks whether, given an NFA, the NFA accepts all strings over its alphabet. If we set the only secret state $q_s$ of the NFA $\A$ of Theorem~\ref{k-so-seth2} to be non-accepting and all the other states to be accepting, we obtain an NFA that is universal if and only if the automaton $\A$ is CSO with respect to $\{q_s\}$ and $Q-\{q_s\}$.
  We thus have the following consequence improving the result of Fernau and Krebs~\cite{FernauK17}.

  \begin{corollary}\label{NFAuniv}
    Unless SETH fails, there is no algorithm deciding whether a given $n$-state NFA is universal that runs in time $O^*(2^{n/(2+\eps)})$, for any $\eps > 0$.\qed
  \end{corollary}

  Consequently, we immediately have the following result.
  \begin{corollary}
    Given two NFAs $\A_1$ and $\A_2$ with $n_1$ and $n_2$ states,
    respectively, let $n=\max(n_1, n_2)$.  Unless SETH fails, there is no algorithm deciding
    whether $L_m(\A_1) \subseteq L_m(\A_2)$ in time
    $O^*(2^{n/(2+\eps)})$, and there is no algorithm deciding
    whether $L_m(\A_1) = L_m(\A_2)$ in time
    $O^*(2^{n/(2+\eps)})$, for any
    $\eps > 0$.\qed
  \end{corollary}

  We left the question whether Theorem~\ref{k-so-seth2} also holds for NFAs over a fixed-size alphabet open. Although we did not answer this question,
  we showed that ETH implies the non-existence of sub-exponential-time algorithms deciding various types of opacity over a binary alphabet.

  Inspecting Table~\ref{table01}, the reader may notice quite a large gap between the lower and upper bounds for the verification of IFO without any restrictions on the form of non-secret pairs. To improve the upper bound or to (conditionally) show that no such improvements are possible is a challenging open problem.

  It is worth noticing that the construction in the proof of Corollary~\ref{cor-seth} produces an instance of a special case of the problem where the non-secret pairs are of the form $IQ_{NS}=I_{NS}\times F_{NS}$, and hence the special case is tight under the strong exponential time hypothesis.

\bibliography{mybib}

\appendix

\section{Proofs of Claims}

  \claimA*
  \begin{claimproof}
    By induction on $\ell$.
    The initial configuration of $\A_\varphi^X$ is $\{x_{n}^0, x_{n-1}^{0}, \ldots, x_1^{0}\}$, where $00\cdots 0$ represents $\ell=0$ in binary.
    Assume that the configuration of $\A_\varphi^X$ after reading the prefix of $w\in Z_n$ of length $\ell < 2^n-1$ is $\{x_{n}^{r_{n}}, x_{n-1}^{r_{n-1}}, \ldots, x_1^{r_1}\}$, where $r_{n}r_{n-1}\cdots r_1$ represents $\ell$ in binary.
    Let $r_t$ be the rightmost zero of $r_{n}r_{n-1}\cdots r_1$; that is, $r_{n}r_{n-1}\cdots r_1 = r_{n}r_{n-1}\cdots r_{t+1} 0 1 \cdots 1$. Then, $\ell+1$ is represented as $r_{n}r_{n-1}\cdots r_{t+1} 1 0 \cdots 0$ in binary, and because it has $t-1$ trailing zeros, the $(\ell+1)$st symbol of $w$ is of the form $\{a_{t}\}\times C$ by~\eqref{key}.
    It remains to show that every $(a_t,c)\in \{a_t\}\times C$ moves $\A_\varphi^X$ from the configuration $\{x_{n}^{r_{n}}, x_{n-1}^{r_{n-1}}, \ldots, x_{t+1}^{r_{t+1}}, x_t^0, x_{t-1}^1, \ldots, x_1^1\}$ to the configuration $\{x_{n}^{r_{n}}, x_{n-1}^{r_{n-1}}, \ldots, x_{t+1}^{r_{t+1}}, x_{t}^1, x_{t-1}^0,\ldots, x_1^{0}\}$.
    By the definition of $\A_\varphi^X$, the transition under $(a_{t},c)$ is undefined in states $x_{t-1}^{1}, \ldots, x_1^{1}$, it is a self-loop in states $x_{n}^{r_{n}}, \ldots, x_{t+1}^{r_{t+1}}$, and it moves $\A_\varphi^X$ from state $x_{t}^0$ to states $x_{t}^1$ and $x_{t-1}^{0}, \ldots, x_{1}^{0}$. Therefore, the automaton $\A_\varphi^X$ moves to the configuration $\{x_{n}^{r_{n}}, x_{n-1}^{r_{n-1}}, \ldots, x_{t+1}^{r_{t+1}}, x_{t}^1, x_{t-1}^0,\ldots, x_1^{0}\}$, where $r_{n}r_{n-1}\cdots r_{t+1} 1 0\cdots {0}$ represents $\ell+1$ in binary.
  \end{claimproof}

  \claimB*
  \begin{claimproof}
    The secret state $q_s$ is an initial state of $\A_\varphi$ and since it contains a self-loop under all symbols of $\Sigma$, it appears in every configuration of $\A_\varphi$.
  \end{claimproof}

  We now show that if there is a satisfying assignment, then the non-secret state $q_{ns}$ is reached by $\A_\varphi$. To this end, for a variable $x$ and $r\in\{0,1\}$, we define the function
  \[
    \textsc{lit}(x^r)=
    \begin{cases}
      \neg x & \text{ if } r=0\\
            x & \text{ if } r=1
    \end{cases}
  \]

  \claimC*
  \begin{claimproof}
    By Claim~\ref{claim3b}, the configuration of $\A_\varphi$ after reading $w$ is $\{x_{n}^{r_{n}}, x_{n-1}^{r_{n-1}},\ldots, x_1^{r_1}\}\cup Y,$ where
    $Y \in \{ \{q_s\}, \{q_s,q_{ns}\} \}$.
    Since the assignment $r_{n}r_{n-1}\cdots r_{1}$ satisfies $\varphi$, for every $c \in C$, there is $i$ such that $\textsc{lit}(x_i^{r_i})$ satisfies $c$, and hence there are transitions from $x_i^{r_i}$ to $q_{ns}$ under $(a_j, c)$ for all $j = 1,\ldots,n+1$. Therefore, after reading any symbol from $\Sigma$, the configuration of $\A_\varphi$ contains $q_{ns}$.
  \end{claimproof}

  \claimD*
  \begin{claimproof}
    We construct a sequence $\eps = w_0, w_1, \dots, w_{2^n - 1} = w_{\varphi}$ of prefixes of the required string such that, for $\ell=0,\ldots,2^n-1$, the configuration of $\A_\varphi$ after reading $w_\ell$ is $\{x_n^{r_n},\dots,x_1^{r_1}\} \cup \{q_s\}$ and $r_n\cdots r_1$ is the representation of $\ell$ in binary.

    We proceed by induction on $\ell$. Since the initial configuration of $\A_\varphi$ is $\{x_n^0,x_{n-1}^0,\dots,x_1^0\} \cup \{q_s\}$, the claim holds for $\ell=0$. Now, assume that the claim holds for $w_\ell$ with $\ell < 2^n-1$, and denote the configuration of $A_\varphi$ after reading $w_\ell$ by $\{x_n^{r_n},\dots,x_1^{r_1}\} \cup \{q_s\}$. We show that there is $(a,c) \in \Sigma$ such that $w_{\ell+1} = w_{\ell}(a,c)$. Let $r_t$ be the rightmost 0 of $r_n\cdots r_1$ and take $a = a_t$. Then, for any choice of $c$, the string $w_{\ell+1}$ is, by~\eqref{key}, a prefix of a string in $Z_n$, and, by Claim~\ref{claim3a}, the configuration of $A_\varphi$ after reading $w_{\ell+1}$ is $\{x_n^{r'_n},\ldots ,x_1^{r'_1}\} \cup Y$, where $Y \in \{\{q_s\}, \{q_s,q_{ns}\}\}$ and $r'_n\cdots r'_1$ represents $\ell+1$ in binary. Since $\varphi$ is not satisfiable, there is a clause $c'$ that is not satisfied by $\textsc{lit}(x_i^{r_i})$ for any $i = 1,\ldots,n$. Taking $c = c'$ then gives the required symbol. Indeed, the transition from $x_i^{r_i}$ under $(a_t,c')$ is undefined for $i = 1,\dots,t-1$, it is a self-loop for $i=t+1,\ldots,n$, and it takes $x_t^0$ to $\{x_t^1\} \cup \{x_j^0 \mid 1 \leq j \leq t-1\}$, see Figure~\ref{fig:kso-seth-cons}; therefore, $Y=\{q_s\}$.
  \end{claimproof}

\end{document}